\documentclass[11pt]{article}

\usepackage{tikz}
\usetikzlibrary{matrix,arrows,shapes}
\usepackage{pgf}
\usepackage{url}
  \usepackage{enumerate}
  \usepackage{amsthm,amssymb,amsmath,amscd}
  \usepackage{eucal}
  \usepackage{mathrsfs}
  \usepackage{mathtools}
  \usepackage{upgreek}
   \usepackage{dsfont}
      \usepackage{yfonts}
      \usepackage{subfigure}
  \usepackage[T1]{fontenc}
  \usepackage{bbm}
  \usepackage{geometry}
\usepackage{array}
\usepackage{booktabs}
\newcommand{\A}{\mathfrak{A}}

\newcommand{\1}{\mathds{1}}

\newcommand{\al}{\alpha}

\newcommand{\eps}{\varepsilon}

\newcommand{\la}{\lambda}

\newcommand{\obj}{\mathrm{Obj}}

\newcommand{\Loc}{\mathrm{\mathbf{Loc}}} 
\newcommand{\TLoc}{\mathrm{\mathbf{Loc}}^{\otimes}}
\newcommand{\Obs}{\mathrm{\mathbf{Obs}}} 
\newcommand{\TObs}{\mathrm{\mathbf{Obs}}^{\otimes}}

\newcommand{\sst}[1]{\scriptscriptstyle{#1}}

\newcommand{\be}{\begin{equation}}
\newcommand{\ee}{\end{equation}}

\newcommand{\NN}{\mathbb{N}} 
\newcommand{\CC}{\mathbb{C}} 

\newcommand{\Lcal}{\mathcal {L}}
\newcommand{\Bcal}{\mathcal {B}}
\newcommand{\Kcal}{\mathcal{K}}

\newcommand{\Hcal}{\mathcal{H}}

\newcommand{\Ncal}{\mathcal{N}}
\newcommand{\Mcal}{\mathcal{M}}
\newcommand{\Ocal}{\mathcal{O}}
\newcommand{\Rcal}{\mathcal{R}}

\newcommand{\mintens}{\otimes_{\sst{\mathrm{min}}}}
 \newtheorem{thm}{Theorem}[section]

 \newtheorem{prop}[thm]{Proposition}
 \theoremstyle{definition}
 \newtheorem{defn}[thm]{Definition}
 \theoremstyle{remark}

 \numberwithin{equation}{section}

\begin{document}

%
%
%
%
%
%
%
%
%

\title{The Locality Axiom in Quantum Field Theory and \\Tensor Products of $C^*$-algebras}
\author{Romeo Brunetti, Klaus Fredenhagen, Paniz Imani, Katarzyna Rejzner}
\maketitle
\begin{abstract}
The prototype of mutually independent systems are systems which are localized in spacelike separated regions. In the framework of locally covariant quantum field theory we show that the commutativity of observables in spacelike separated regions can be encoded in the tensorial structure of the functor which associates unital $C^*$-algebras (the local observable algebras) to globally hyperbolic spacetimes. This holds under the assumption that the local algebras satisfy the split property and involves the minimal tensor product of $C^*$-algebras. 
\end{abstract}
\section{Introduction}
One of the crucial properties of properties of quantum field theory is the commutativity of spacelike localized observables, a property often termed {\it locality} or {\it Einstein causality}. It was shown by Roos \cite{Roos} that this property implies that the algebra generated by the observables in two spacelike separated regions is isomorphic to a tensor product of the algebras associated to these regions. Unfortunately, $C^*$-algebras do not possess, in general, unique tensor products (see, e.g. \cite{Takesaki}), so it remained open which tensor product occurs in quantum field theory.

A partial answer can be given when the theory satisfies the so-called split property. The split property states that between algebras of regions such that the closure of the smaller is contained in the interior of the larger region there exists an intermediate type I factor, i.e. an algebra isomorphic to the algebra of all bounded operators on a Hilbert space. This property was conjectured by Borchers, first proven by Buchholz for the free scalar field \cite{Buchholz} and in depth analyzed  by Doplicher and Longo \cite{DL}. Under this assumption it follows that the algebra generated by observables of regions whose closures are spacelike separated is the minimal tensor product.

After generalizing the Haag-Kastler nets of local observables associated to subregions of Minkowski space \cite{HK} to a functor from globally hyperbolic spacetimes to algebras (the concept of locally covariant quantum field theory \cite{BFV,HW}) it soon became clear that the axiom of Einstein causality is related to a tensor structure of the functor \cite{BF,FR}. On the category of spacetimes, with causality preserving isometric embeddings as morphisms, the tensor product is just the disjoint union, and morphisms are embeddings where the connected components are mapped to spacelike separated regions. It was, however, not clear how to choose the tensor structure on the category of $C^*$-algebras. A related problem occurs in perturbative algebraic quantum field theory. But there one replaces the $C^*$-algebras by locally convex topological algebras which turn out to be nuclear and therefore have a unique tensor product \cite{BFL-R}.

In the present paper we show that under the assumption that the split property holds for connected spacetimes, the causality axiom is equivalent to the tensorial property of the functor extended to disconnected spacetimes. Here the category of $C^*$-algebra is equipped with the minimal (or spatial) tensor product.  
\section{Locally covariant quantum field theory}
The framework of locally covariant quantum field theory was developed in \cite{BFV,HW}. It allows to 
construct the theory simultaneously on all spacetimes of a given class in a coherent way. To this end it is convenient to use the language of category theory. Therefore, let us now define the categories that will be used in this work. The first one, denoted by $\Loc$ is related to the notion of physical spacetimes and their embeddings. Its objects are all n-dimensional ($n\geq2$ is fixed) spacetimes\footnote{By spacetime we mean a smooth Hausdorff, paracompact and connected manifold.} $\Mcal\doteq(M,g)$, which are globally hyperbolic, oriented and time-oriented\footnote{The assumptions are adapted to the case of scalar fields. In more general cases one might be forced to impose more restrictive conditions. For detailed discussion see e.g. \cite{DaLa}.} and  $g$ is a smooth Lorentzian metric. Morphisms of $\Loc$ are defined as
 isometric orientation and time-orientation preserving embeddings that are additionally causality preserving, i.e. they fulfill the following condition: given $(M_1,g_1),(M_2,g_2)\in\obj(\Loc)$, for any causal curve $\gamma : [a,b]\to M_2$, if $\gamma(a),\gamma(b)\in\chi(M_1)$ then for all $t	\in ]a,b[$ we have: $\gamma(t)\in\chi(M_1)$. We will call an embedding satisfying these criteria \emph{admissible}.
 
In local quantum physics one assigns to regions of spacetime algebras of observables. To formulate this assignment in the category theory language, we need to define the category of observables $\Obs$. In quantum theory it can be chosen as the category of unital  $C^*$-algebras with injective $*$-homomorphisms as morphisms. A locally covariant QFT is defined by assigning to spacetimes  $\Mcal$ corresponding unital $C^*$-algebras $\A(\Mcal)$. This assignment has to fulfill a set of axioms, which generalize the Haag-Kastler axioms:
\begin{enumerate}
\item	If $\chi: \Ncal \rightarrow \Mcal$ is an admissible embedding, then $\alpha_\chi:\ \A(\Ncal)\rightarrow \A(\Mcal)$ is a unit preserving injective $C^*$-homomorphism (\textit{subsystems}),
\item	Let $\chi:\Ncal \rightarrow \Mcal$, $\chi':\Mcal\rightarrow \Lcal$ be admissible embeddings, then
$\alpha_{\chi'\circ\chi} = \alpha_{\chi'}\circ\alpha_\chi$ (\textit{covariance}),
\item If $\chi_1 : \Ncal_1 \rightarrow \Mcal$, $\chi_2 : \Ncal_2 \rightarrow \Mcal$ are admissible embeddings such that $\chi_1(\Ncal_1)$ and $\chi_2(\Ncal_2)$ are spacelike separated in $\Mcal$ then
$[\alpha_{\chi_1} (\A(\Ncal_1)), \alpha_{\chi_2} (\A(\Ncal_2))] = {0}$ (\textit{Einstein causality}),
\item	If $\chi(\Ncal)$ contains a Cauchy surface of $\Mcal$ then $\alpha_\chi(\A(\Ncal )) = \A(\Mcal)$ (\textit{timeslice axiom}).
\end{enumerate}
Axioms 1 and 2 mean simply that $\A$ is a covariant functor $\A$ between $\Loc$ and $\Obs$, with $\A\chi:=\alpha_\chi$. 
The third axiom is related to the tensorial structure of the underlying categories and we will now focus on its categorical formulation.
We would like to extend  $\Loc$ and $\Obs$ to tensor categories. By this we mean strictly monoidal categories. Following \cite{MacLane}
we call a category $\mathbf{C}$ strictly monoidal if there exists a bifunctor $\otimes:\mathbf{C}\times\mathbf{C}\rightarrow\mathbf{C}$ which is associative, i.e. $\otimes(\otimes\times 1)=\otimes( 1\times\otimes)$ and there exists an object $e$ which is a left and right unit for $\otimes$.

The category $\Loc$  can be extended to a tensor category if we drop the condition of connectedness and define $\TLoc$ as the category whose objects are 
finite disjoint unions of elements of $\obj(\Loc)$, \[\Mcal=\Mcal_1\sqcup \ldots\sqcup \Mcal_N\] where $\Mcal_i\in\obj(\Loc)$. The tensor product is defined as the disjoint union $\otimes\doteq \sqcup$
and the unit $e$ is the empty set $\varnothing$. The admissible embeddings are maps $\chi: \Mcal_1\sqcup \ldots\sqcup \Mcal_n\rightarrow  \Mcal$ such that each component satisfies the requirements for a morphism of $\Loc$ and additionally all images are spacelike to each other, i.e., $\chi(\Mcal_1) \perp\ldots\perp\chi(\Mcal_n)$. 

Now we want to extend the category $\Obs$ to a tensor category. Since there is no unique tensor structure on general locally convex vector spaces, one has to make a  choice of the tensor structure basing on some physical requirements. In particular one can define the minimal and the maximal $C^*$-tensor norm. The minimal norm $\|.\|_{\textrm{min}}$ is defined as:
\[
\|A\|_{\textrm{min}} \doteq \sup\{\|(\pi_1\otimes\pi_2)(A)\|_{\Bcal(\Hcal_1\otimes\Hcal_2)}\}\ , \quad A\in\A_1\otimes\A_2\,,
\]
where $\pi_1$ and $\pi_2$ run through all representations of $\A_1$ and of $\A_2$ on Hilbert spaces $\Hcal_1$, $\Hcal_2$ respectively. $\Bcal$ denotes the algebra of bounded operators. If we choose $\pi_1$ and $\pi_2$ to be faithful, then the
supremum is achieved, i.e. $\|A\|_{\textrm{min}}=\|(\pi_1\otimes\pi_2)(A)\|_{\Bcal(\Hcal_1\otimes\Hcal_2)}$.
The completion of the algebraic tensor product $\A_1\otimes\A_2$  with respect to the minimal norm $\|A\|_{\textrm{min}}$ is a $C^*$-algebra and is denoted by $\A_1\underset{\sst{\textrm{min}}}{\otimes}\A_2$. 

We will argue now that, under some further assumptions concerning the QFT functor $\A$, the minimal $C^*$-tensor product is a natural choice for introducing a tensor structure on $\Obs$. Firstly, we assume $\A$ to be additive in the following sense:
\begin{defn}\label{additiv}
Let $\Ocal\in\obj(\Loc)$ and let $\{\chi_i\in\hom(\Ocal_i,\Ocal)\}_{i\in I}$ be a family of morphisms such that $\overline{\chi_i(\Ocal_i)}\subset \Ocal$ are compact, $\Ocal\subset\bigcup\limits_{i\in I}\chi_i(\Ocal_i)$ and it holds $\overline{\chi_i(\Ocal_i)}\subset \chi_{i+1}(\Ocal_{i+1})$ (i.e. we have an increasing family of subsets covering $\Ocal$). We say that $\A$ is additive if 
\[
\A(\Ocal)=\overline{\bigcup\limits_{i\in I}\al_{\chi_i}(\A(\Ocal_i))}\,.
\]
holds for all such families, where the bar denotes the closure in the norm topology.
\end{defn}

Moreover we require  $\A$ to satisfy the \textit{split property}\footnote{The split property was originally defined for inclusions of von Neumann algebras. Here we introduce a corresponding notion, which is adapted to the locally covariant setting. A weaker condition, called \textit{intermediate factoriality}, was already discussed in \cite{BFV}.}. 
\begin{defn}\label{split}
We say that $\A$ satisfies the split property if the following two conditions are satisfied:
\begin{enumerate}
\item For every two spacetimes
$\Mcal\xhookrightarrow[\chi]{}\Ncal$ in $\obj(\Loc)$, with compact and connected closure of the image $\overline{\chi(M)}$, contained in the interior of $N$, there exist a type I von Neumann factor $\Rcal$ such that
\[
\al_\chi(\A(\Mcal))\subset \Rcal\subset \A(\Ncal)\,,
\]
\item If $\psi\in\hom(\Ncal,\Lcal)$, $\Lcal\in\obj(\Loc)$, then
for every increasing net of elements $A_i\in\Rcal$ such that $A_i\nearrow A$, it holds
\[
\sup_i \al_\psi(A_i)=\al_\psi(A)\,.
\]
\end{enumerate}
\end{defn}
The second condition guarantees that in the sequence below the inclusion $\al_\psi:\Rcal\rightarrow \tilde{\Rcal}$ is $\sigma$-continuous
\[
\al_\psi\circ \al_\chi(\A(\Mcal))\subset \al_\psi(\Rcal)\subset  \al_\psi(\A(\Ncal))\subset\tilde{\Rcal}\subset\A(\Lcal)\,.
\]
In other words, given a normal state $\tilde{\Rcal}_*\ni\omega$ one obtains a normal state in $\Rcal_*$ by taking $\al^*_\psi\omega\upharpoonright_\Rcal$.

From now on we assume that properties \ref{split} and \ref{additiv} hold for the QFT functor $\A$. Let us now fix a spacetime $\Mcal=(M,g)\in\obj(\Loc)$ and consider a family $\Kcal(\Mcal)$ of subsets of $M$ which are relatively compact and are objects of $\Loc$, when equipped with $g_{\Ocal}$, the restriction of the metric $g$. The inclusion map $\iota_{M,\Ocal}:\Ocal\rightarrow M$ is a morphism in $\hom((\Ocal,g_{\Ocal}),(M,g))$. We denote $\al_{M,\Ocal}\doteq \A\iota_{M,\Ocal}$ and 
\[
\A_{\Mcal}(\Ocal)\doteq\al_{M,\Ocal}(\A((\Ocal,g_{\Ocal})))\,.
\]
It was shown in \cite{BFV} that if $\A$ is a covariant functor satisfying in addition axioms 3 and 4, then $\{\A_{\Mcal}(\Ocal)\}_{\Ocal\in\Kcal(\Mcal)}$ is a Haag-Kastler net on the spacetime $\Mcal$. It is also easy to see that if $\A$ satisfies \ref{split} and \ref{additiv}, then additivity and split property hold for the corresponding Haag-Kastler net. With these two properties we can write a local algebra $\A_{\Mcal}(\Ocal)$ as an inductive limit of type I factors, since the additivity implies that
\[
\A_{\Mcal}(\Ocal)=\overline{\bigcup\limits_{i\in I}\A_{\Mcal}(\Ocal_i)}\,,
\]
where $\{\Ocal_i\}_{i\in I}$ is an increasing family of relatively compact, contractible, causally convex subsets of $M$ such that $\overline{\Ocal_i}\subset\Ocal_{i+1}$ (we will denote such inclusion as $\Ocal_i\Subset\Ocal_{i+1})$ and $\Ocal\subset\bigcup\limits_{i\in I}\Ocal_i$. From the split property it follows that we can write this inductive limit with type I factors
\[
\A(\Ocal)=\overline{\bigcup\limits_{i\in I}\Rcal^i}\,,
\]
where $\A_{\Mcal}(\Ocal_i)\subset\Rcal_i\subset\A_{\Mcal}(\Ocal_{i+1})$.

Let us now consider two local algebras associated with spacelike separated regions $\Ocal_1, \Ocal_2\subset M$. 
The split property allows us to write them as inductive limits of increasing families of type I factors: $\A_{\Mcal}(\Ocal_1)=\overline{\bigcup\limits_{i\in I}\Rcal_1^i}$, $\A_{\Mcal}(\Ocal_2)=\overline{\bigcup\limits_{j\in J}\Rcal_2^j}$. Let $\Ocal$ be a region of $M$ containing both  $\Ocal_1$ and $\Ocal_2$. We will prove now that the operator norm 
on the algebra generated by $\A_{\Mcal}(\Ocal_1)$ and $\A_{\Mcal}(\Ocal_2)$
coincides with the minimal $C^*$-tensor norm of 
$\A_{\Mcal}(\Ocal_1)\underset{\sst{\textrm{min}}}{\otimes}\A_{\Mcal}(\Ocal_2)$.
\begin{thm}\label{minimal}
Let $\A$ be a QFT functor 
which is additive and satisfies the split property,
and let $\A_{\Mcal}$ be the induced local net on the spacetime $\Mcal\in\obj(\Loc)$. For spacelike separated regions $\Ocal_1$, $\Ocal_2$ contained in $\Ocal\subset M$\footnote{Note that $\Ocal_1$ and $\Ocal_2$ are not assumed to be relatively compact.} the 
map 
\[ \alpha:\left\{\begin{array}{ccc}
                     \A_{\Mcal}(\Ocal_1)\otimes\A_{\Mcal}(\Ocal_2)&\to& \A_{\Mcal}(\Ocal)\\
                     A_1\otimes A_2&\mapsto & A_1A_2 
                     \end{array}\right.\]
is isometric with respect to the minimal norm, $\|\alpha(A)\|=\|A\|_{\sst{\mathrm{min}}}$.
\end{thm}
\begin{proof}  
Using the additivity and split property we can write  generic elements $A_n$, $B_n$  $n=1,\dots,N$ ($N\in\NN$) of $\A_{\Mcal}(\Ocal_1)$ and $\A_{\Mcal}(\Ocal_2)$ respectively as  $A_n=\lim\limits_{i\in I} A^i_n$, $B_n=\lim\limits_{j\in J} B_n^j$, where $ A^i_n\in\Rcal_{1,i}$, $B^j_n\in\Rcal_{2,j}$ and $\Rcal_{1,i}$, $\Rcal_{2,j}$ are type I factors. Since the algebras $\A_{\Mcal}(\Ocal_1)$ and $\A_{\Mcal}(\Ocal_2)$ commute with each other, a generic element of $\A_{\Mcal}(\Ocal_1)\vee\A_{\Mcal}(\Ocal_2)$ can be written as $\sum_{n=1}^N A_nB_n$. To prove the theorem we have to show that
\be\label{min:tensor:norm}
\|\sum_{n=1}^N A_nB_n\|=\|\sum_{n=1}^NA_n \otimes B_n\|_{\textrm{min}}\,.
\ee
Let us choose $A^i_n$,  $B^j_n$ such that $\|\sum\limits_{n=1}^N  A_nB_n-\sum\limits_{n=1}^N  A_n^iB_n^j\|<\eps$ and
 $\|\sum\limits_{n=1}^N  A_n\otimes B_n-\sum\limits_{n=1}^N  A_n^i\otimes B_n^j\|_{\textrm{min}}<\eps$. It follows now that:
\[
-2\eps\leq\|\sum_{n=1}^N A_nB_n\|-\|\sum_{n=1}^N A^i_nB^j_n\|+\|\sum_{n=1}^NA^i_n \otimes B^j_n\|_{\textrm{min}}-\|\sum_{n=1}^NA_n \otimes B_n\|_{\textrm{min}}\leq 2\eps\,.
\] 
To prove \eqref{min:tensor:norm} it remains to prove that the minimal $C^*$ tensor norm $\|.\|_{\textrm{min}}$ coincides for elements $A^i_n\in\Rcal_{1,i}$, $B^j_n\in\Rcal_{2,j}$ with the norm on $\Rcal_{ij}$. Let us consider the following inclusions of subsets of $\Ocal$:
\[
\begin{tikzpicture} \matrix(a)[matrix of math nodes, row sep=1em, column sep=0.3em, text height=2.5ex, text depth=0.2ex, row 1 column 5/.style={blue!80!black}, row 3 column 5/.style={blue!80!black}, row 2 column 7/.style={blue!80!black}] {
 \Ocal_{1,i}          &\Subset          &\Ocal_{1,i+1}      &\Subset    &\Ocal_1        &            &\\
 			    &			   &				 &		   &\tilde{\Ocal}&\Subset&\Ocal\\
\Ocal_{2,j}           &\Subset          &\Ocal_{2,j+1}      &\Subset    &\Ocal_2        &            &\\
};
 \path[-,color=white](a-1-3) edge node[midway, sloped, color=black]{$\subset$} (a-2-5); 
 \path[-,color=white](a-1-5) edge node[above, sloped, color=blue!80!black]{$\subset$} (a-2-7); 
 \path[-,color=white](a-3-5) edge node[below, sloped, color=blue!80!black]{$\subset$} (a-2-7); 
 \path[-,color=white](a-3-3) edge node[midway, sloped, color=black]{$\subset$} (a-2-5); 
\end{tikzpicture}
\]
The existence of an appropriate $\tilde{\Ocal}$ is guaranteed by the fact that $\overline{\Ocal_{1,i+1}}\subset\Ocal_1\subset \Ocal$. One can illustrate the above sequence of inclusions on a picture, see figure \ref{inclusions:sets}.
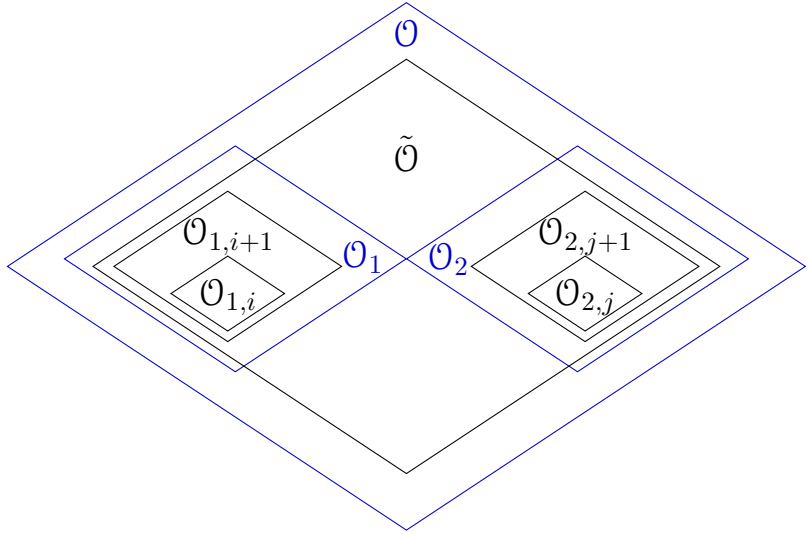
\begin{figure}[tb]
\begin{center}
 \begin{tikzpicture}
\draw (-0.1,-0.36) node[diamond,draw,minimum width=1.5cm,minimum height=1cm] {};
\draw (-0.1,0) node[diamond,draw,minimum width=3cm,minimum height=2cm] {};
\draw (4.6,-0.36) node[diamond,draw,minimum width=1.5cm,minimum height=1cm] {};
\draw (4.6,0) node[diamond,draw,minimum width=3cm,minimum height=2cm] {};
\draw (4.5,0.1) node[diamond,draw,minimum width=4.5cm,minimum height=3cm,color=blue!80!black] {};
\draw (0,0.1) node[diamond,draw,minimum width=4.5cm,minimum height=3cm,color=blue!80!black] {};\draw(-0.1,-0.4) node {\Large{$\Ocal_{1,i}$}};
\draw(-0.1,0.4) node {\Large{$\Ocal_{1,i+1}$}};
\draw(4.6,-0.4) node {\Large{$\Ocal_{2,j}$}};
\draw(4.6,0.4) node {\Large{$\Ocal_{2,j+1}$}};
\draw(1.67,0.13) node [color=blue!80!black]{\Large{$\Ocal_{1}$}};
\draw(2.8,0.1) node [color=blue!80!black]{\Large{$\Ocal_{2}$}};
\draw (2.25,0) node[diamond,draw,minimum width=8.25cm,minimum height=5.5cm] {};
\draw(2.25,1.5) node {\Large{$\tilde{\Ocal}$}};
\draw (2.25,0) node[diamond,draw,minimum width=10.5cm,minimum height=7cm,color=blue!80!black] {};
\draw(2.25,3.1) node[color=blue!80!black] {\Large{$\Ocal$}};
\end{tikzpicture}
\end{center}
\caption{Spacetime configuration of open sets used in the proof.\label{inclusions:sets}}
\end{figure}
To inclusions of open subsets of $\Mcal$ there correspond the following inclusions of algebras:
\[
\begin{tikzpicture} \matrix(a)[matrix of math nodes, row sep=1em, column sep=0.3em, text height=2.5ex, text depth=0.2ex, row 1 column 7/.style={blue!80!black}, row 3 column 7/.style={blue!80!black}, row 2 column 11/.style={blue!80!black}, row 1 column 3/.style={red!70!black}, row 3 column 3/.style={red!70!black}, row 2 column 9/.style={red!70!black}] {
\A_{\Mcal}(\Ocal_{1,i})&\subset &\Rcal_{1,i}&\subset   &\A_{\Mcal}(\Ocal_{1,i+1})&\subset &\A_{\Mcal}(\Ocal_1)        &             &   &           &\\
 			               & 	     &			&		&		 			     &		   &\A_{\Mcal}(\tilde{\Ocal})&\subset&\Rcal_{ij}&\subset&\A_{\Mcal}(\Ocal)\\
\A_{\Mcal}(\Ocal_{2,j}) &\subset &\Rcal_{2,j}&\subset  &\A_{\Mcal}(\Ocal_{2,j+1})&\subset &\A_{\Mcal}(\Ocal_2)        &               &  &           &\\
};
 \path[-,color=white](a-1-5) edge node[midway, sloped, color=black]{$\subset$} (a-2-7); 
  \path[-,color=white](a-3-5) edge node[midway, sloped, color=black]{$\subset$} (a-2-7); 
 \path[-,color=white](a-1-7) edge node[above, sloped, color=blue!80!black]{$\subset$} (a-2-11); 
 \path[-,color=white](a-3-7) edge node[below, sloped, color=blue!80!black]{$\subset$} (a-2-11); 
\end{tikzpicture}
\]
The second condition in the definition of the split property \ref{split} ensures that inclusions of factors $\al_{\Ocal,\Ocal_{1,i+1}}:\Rcal_{1,i}\rightarrow\Rcal_{ij}$,  $\al_{\Ocal,\Ocal_{2,j+1}}:\Rcal_{2,j}\rightarrow\Rcal_{ij}$ are  $\sigma$-continuous. Moreover $\al_{\Ocal,\Ocal_{1,i+1}}(\Rcal_{1,i})$ and $\al_{\Ocal,\Ocal_{2,j+1}}(\Rcal_{2,j})$ commute in $\Rcal_{ij}$, so we can conclude that the minimal $C^*$-norm $\|.\|_{\textrm{min}}$ coincides for elements $A^i_n\in\Rcal_{1,i}$, $B^j_n\in\Rcal_{2,j}$ with the norm in $\A_{\Mcal}(\Ocal)$. Therefore it holds:
 \[
-2\eps\leq \|\sum_{n=1}^N A_nB_n\|-\|\sum_{n=1}^NA_n \otimes B_n\|_{\textrm{min}}\leq 2\eps\,,
 \]
 and since $\eps$ can be chosen arbitrarily small, it follows that \eqref{min:tensor:norm} holds.
\end{proof}

The above argument shows that $\mintens$ is a choice of tensor product compatible with the net structure. We can now prove the following:
\begin{prop}
We can extend $\Obs$ to a tensor category $\TObs$ by defining the tensor functor $\otimes$ as:
\begin{align*}
\A_1\otimes\A_2&\doteq \A_1\mintens\A_2\,,\\
(\al_1\otimes\al_2)(A\otimes B)&\doteq \al_1(A)\otimes\al_2(B)\,,A\in\A_1,B\in\A_2\,,
\end{align*}
where $A\in\A_1$, $B\in\A_2$, $\A_1,\A_2\in\obj(\TObs)$, $\al_1\in\hom(\A_1,\tilde{\A}_1)$, $\al_2\in\hom(\A_2,\tilde{\A_2})$, the class of objects is the same for $\TObs$ as for $\Obs$ and the unit of $\TObs$ is provided by $e\doteq\CC$.
\end{prop}
\begin{proof}
The only nontrivial step is to check that $\al_1\otimes\al_2$ is an isometric embedding. 
This follows from the injectivity of the minimal $C^*$-norm. Let $\pi_1$, $\pi_2$ be faithful representations of $\tilde{\A}_1$, $\tilde{\A}_2$ respectively on Hilbert spaces $\Hcal_1$, $\Hcal_2$. From the definition of the minimal tensor norm it follows that for an element of the algebraic tensor product $\sum_{n=1}^NA_n\otimes B_n$ we obtain:
\be\label{isometry1}
\Big\|(\al_1\otimes\al_2)\Big(\sum_{n=1}^NA_n\otimes B_n\Big)\Big\|_{\mathrm{min}}=\Big\|\sum_{n=1}^N\pi_1(\al_1(A_n))\otimes\pi_2(\al_2(B_n)) \Big\|_{\Bcal(\Hcal_1\otimes\Hcal_2)}\,.
\ee
On the other hand, since $\al_1$, $\al_2$ are isometric embeddings, $\pi_1\circ\al_1$ and $\pi_2\circ\al_2$ are faithful representations of $\A_1$ and $\A_2$ respectively. Therefore
\be\label{isometry2}
\Big\|\sum_{n=1}^NA_n\otimes B_n\Big\|_{\mathrm{min}}=\Big\|\sum_{n=1}^N\pi_1\circ\al_1(A_n)\otimes\pi_2\circ\al_2(B_n) \Big\|_{\Bcal(\Hcal_1\otimes\Hcal_2)}\,.
\ee
From \eqref{isometry1} and  \eqref{isometry2} it follows that $\al_1\otimes\al_2$ defined as above is an isometry, so in particular it is continuous and extends to arbitrary elements of the completion $\A_1\mintens\A_2$.
\end{proof}
The QFT functor $\A$ can now be extended to a functor $\A^\otimes$ between the categories $\Loc^\otimes$ and $\Obs^\otimes$. It is a covariant tensor functor if it holds:
\begin{eqnarray}
\A^\otimes\left(\Mcal_1\sqcup\Mcal_2\right)&=&\A(\Mcal_1)\otimes\A(\Mcal_2)\\
\A^\otimes(\chi\otimes\chi')&=&\A^\otimes(\chi)\otimes\A^\otimes(\chi')\\
\A^\otimes(\varnothing)&=&\CC
\end{eqnarray}
It was already mentioned in \cite{FR} that if $\A$ can be extended to a tensor functor, then the causality follows. Here we give a complete proof of this theorem.
\begin{thm}
Let $\TLoc$ and $\TObs$  be tensor categories of spacetimes and observables, respectively. Then the QFT functor $\A:\Loc\rightarrow\Obs$ is causal if and only if it can be extended to a tensor functor $\A^\otimes:\TLoc\rightarrow\TObs$.
\end{thm}
\begin{proof}
To prove the ``$\Leftarrow$'' implication consider the natural embeddings $\iota_{i}:\Mcal_i\rightarrow\Mcal_1\sqcup\Mcal_2$, $i=1,2$ for which $\A\iota_1(A_1)=A_1\otimes\1$, $\A\iota_2(A_2)=\1\otimes A_2$, $A_i\in\A(\Mcal_i)$. Now let $\chi_i:\Mcal_i\rightarrow\Mcal$ be admissible embeddings such that the images of  $\chi_1$ and  $\chi_2$ are causally disjoint in $\Mcal$. We define now an admissible embedding $\chi:\Mcal_1\sqcup\Mcal_2\rightarrow\Mcal$ as:
\be
\chi(x)=\left\{
\begin{array}{lcl}
\chi_1(x)&,&x\in\Mcal_1\\
\chi_2(x)&,&x\in\Mcal_2
\end{array}
\right.
\ee
Since $\A^\otimes$ is a covariant tensor functor, it follows:
\be
[\A\chi_1(A_1),\A\chi_2(A_2)]=\A\chi[\A\iota_1(A_1),\A\iota_2(A_2)]=\A\chi[A_1\otimes\1,\1\otimes A_2]=0
\ee
This proves the causality. Now we prove the ``$\Rightarrow$'' implication. We have to define the functor $\A^\otimes$. Its action on objects is straightforward, since we can set
\[
\A^\otimes\left(\bigsqcup\limits_k\Mcal_k\right)\doteq\bigotimes\limits_k \A(\Mcal_k)\,.
\]
Now let us consider an admissible embedding $\chi:\Mcal=\bigsqcup\limits_k\Mcal_k\rightarrow\bigsqcup\limits_l\Ncal_l=\Ncal$. Let $\iota_k$ denote the inclusion maps $\iota_k: \Mcal_k\rightarrow \Mcal$ and similarly we write $\rho_l$ for maps  $\Ncal_l\rightarrow \Ncal$. Clearly $\chi\circ \iota_k$ is admissible for every $k$ and moreover there exists an $l$ such that $\chi(\Mcal_k)\subset\Ncal_l$ and the following diagram commutes:
\[
\begin{CD}
\Mcal_k @>\chi_k>>\Ncal_l\\ 
@V{\iota_k}VV     @VV{\rho_l}V\\
\Mcal @>{\chi}>>\Ncal
\end{CD}
\]
where $\chi_k\doteq\chi\upharpoonright_{\Mcal_k}$ is admissible in $\Loc$.
Let $\kappa_l\doteq\{k|\chi(\Mcal_k)\subset\Ncal_l\}$.  Since $\chi$ is admissible, the regions $\chi_k(\Mcal_k)$, $k\in\kappa_l$ are
causally disjoint in $\Ncal_l$. Let us now set:
\be\label{tensor:functor}
\A^\otimes\chi\left(\bigotimes_kA_k\right)\doteq\bigotimes\limits_l\prod\limits_{k\in\kappa_l}\A\chi_k(A_k)\,.
\ee
Since $A_k$, $k\in\kappa_l$ belong to algebras associated with causally disjoined regions, they commute and the above expression is a well defined algebraic isomorphism from $\bigotimes\limits_k \A(\Mcal_k)$ to $\bigotimes\limits_l \A(\Ncal_l)$. Moreover from proposition \ref{minimal} follows that the algebraic isomorphisms $\bigotimes\limits_{k\in\kappa_l}A_k\rightarrow\prod\limits_{k\in\kappa_l} A_k$ are continuous, so we can conclude that $\A^\otimes$ is a morphism in $\TObs$.

Now we have to check the covariance. Let $\chi':\bigsqcup\limits_l\Ncal_l\rightarrow\bigsqcup\limits_j\Lcal_j=\Lcal$ be a morphism of $\TLoc$. We define $\lambda_j\doteq\{l|\chi'(\Ncal_l)\subset\Lcal_j\}$. From the fact that both $\chi$ and $\chi'$ are admissible we can conclude that both squares in the below diagram commute:
\[
\begin{CD}
\Mcal_k @>\chi_k>>\Ncal_l@>\chi'_l>>\Lcal_j\\ 
@V{\iota_k}VV     @VV{\rho_l}V @VV{\psi_j}V\\
\Mcal @>{\chi}>>\Ncal @>{\chi'}>>\Lcal
\end{CD}
\]
It follows now that $\psi_j\circ\chi'_l\circ\chi_k=\chi'\circ\chi\circ\iota_k$ and therefore $(\chi'\circ\chi)_k=\chi'_l\circ\chi_k$. Hence
\begin{align*}
\A(\chi'\circ\chi)\left(\bigotimes_kA_k\right)&=\bigotimes\limits_j\prod\limits_{k\in\bigcup\limits_{l\in\la_j}\kappa_l}\A(\chi'\circ\chi)_k(A_k)\\
&=\bigotimes\limits_j\prod\limits_{l\in\la_j}\A\chi_l'\left(\bigotimes\limits_{k\in\kappa_l}\A\chi_k(A_k)\right)\\&=\A\chi'\left(\bigotimes\limits_l\prod\limits_{k\in\kappa_l}\A\chi_k(A_k)\right)\\
&=\A\chi'\circ\A\chi\left(\bigotimes\limits_k(A_k)\right)\,.
\end{align*}
This shows the covariance and ends the proof that \eqref{tensor:functor} indeed provides an extension of $\A$ to a tensor functor.
\end{proof}
The functor $\A^\otimes$ constructed above is covariant but it will not satisfy the split property, if in definition \ref{split} we allow the image of $\Mcal\xhookrightarrow[\chi]{}\Ncal$ to be disconnected. To fix this one could use a different definition of the tensor category. We would take the von Neumann tensor product $\boxtimes$ to define the algebra of the disjoint union of spacetimes $\Mcal_1\sqcup\Mcal_2$ as:
\[
\A^\boxtimes\left(\Mcal_1\sqcup\Mcal_2\right)\doteq \overline{\bigcup\limits_{i\in I}\Rcal_{1,i}\boxtimes\Rcal_{2,i}}\,,
\]
where we use the split property and additivity to write $\A(\Mcal_1)=\overline{\bigcup\limits_{i\in I}\Rcal_{1,i}}$,  $\A(\Mcal_2)=\overline{\bigcup\limits_{i\in I}\Rcal_{2,i}}$. The increasing families $\{\Ocal_{1,i}\}_{i\in I}$,  $\{\Ocal_{2,i}\}_{i\in I}$ and type I factors $\Rcal_{1,i}$, $\Rcal_{2,i}$  are chosen like in the proof of theorem \ref{minimal}. The objects in the category of observables would not be $C^*$-algebras, but increasing families of type I factors. This formulation seems to be less natural and more technically involved as the one with the minimal $C^*$ -tensor product. An advantage of this approach is that $\A^\boxtimes$ has again the split property.
\end{document}